\pdfoutput=1
\documentclass[11pt,english]{article}
\usepackage[T1]{fontenc}
\usepackage[latin9]{inputenc}
\usepackage{geometry}
\usepackage{babel}
\usepackage{mathtools}
\usepackage{amsmath}
\usepackage{amsthm}
\usepackage{amssymb}
\usepackage{graphicx}

\makeatletter

\newcommand{\noun}[1]{\textsc{#1}}

\usepackage[inline]{enumitem}		
\theoremstyle{plain}
\newtheorem{thm}{\protect\theoremname}
  \theoremstyle{definition}
  \newtheorem{defn}[thm]{\protect\definitionname}
  \theoremstyle{plain}
  \newtheorem{prop}[thm]{\protect\propositionname}
  \theoremstyle{remark}
  \newtheorem{rem}[thm]{\protect\remarkname}
  \theoremstyle{plain}
  \newtheorem{lem}[thm]{\protect\lemmaname}
  \theoremstyle{plain}
  \newtheorem{cor}[thm]{\protect\corollaryname}

\makeatother

  \providecommand{\corollaryname}{Corollary}
  \providecommand{\definitionname}{Definition}
  \providecommand{\lemmaname}{Lemma}
  \providecommand{\propositionname}{Proposition}
  \providecommand{\remarkname}{Remark}
\providecommand{\theoremname}{Theorem}

\begin{document}
\global\long\def\mymin#1{\hphantom{\max_{#1}}\mathllap{\min_{#1}}}

\global\long\def\mymax#1{\hphantom{\max_{#1}}\mathllap{\max_{#1}}}

\global\long\def\supp{\mathrm{supp}}

\title{The Complexity of Simulation and Matrix Multiplication}

\author{Massimo Cairo\thanks{This work was supported by a joint PhD program with the University
of Verona, Department of Computer Science, under PhD grant \textquotedblleft Computational
Mathematics and Biology\textquotedblright .}\\
Università di Trento\\
\texttt{massimo.cairo@unitn.it}\\
\and Romeo Rizzi\\
Università di Verona\\
\texttt{romeo.rizzi@univr.it}}
\maketitle
\begin{abstract}
Computing the simulation preorder of a given Kripke structure (i.e.,
a directed graph with $n$ labeled vertices) has crucial applications
in model checking of temporal logic. It amounts to solving a specific
two-players reachability game, called simulation game. We offer the
first conditional lower bounds for this problem, and we relate its
complexity (for computation, verification, and certification) to some
variants of $n\times n$ matrix multiplication.

We show that any $O(n^{\alpha})$-time algorithm for simulation games,
even restricting to acyclic games/structures, can be used to compute
$n\times n$ boolean matrix multiplication (BMM) in $O(n^{\alpha})$
time. This is the first evidence that improving the existing $O(n^{3})$-time
solutions may be difficult, without resorting to fast matrix multiplication.
In the acyclic case, we match this lower bound presenting the first
subcubic algorithm, based on fast BMM, and running in $n^{\omega+o(1)}$
time (where $\omega<2.376$ is the exponent of matrix multiplication).

For both acyclic and cyclic structures, we point out the existence
of natural and canonical $O(n^{2})$-size certificates, that can be
verified in truly subcubic time. In the acyclic case, $O(n^{2})$
time is sufficient, employing standard matrix product verification.
In the cyclic case, a $\max$-semi-boolean matrix multiplication (MSBMM)
is used, i.e., a matrix multiplication on the semi-ring $(\max,\times)$
where one matrix contains only $0$'s and $1$'s. This MSBMM is computable
(hence verifiable) in truly subcubic $n^{(3+\omega)/2+o(1)}$ time
by reduction to $(\max,\min)$-multiplication.

Finally, we show a reduction from MSBMM to cyclic simulation games
which implies a separation between the cyclic and the acyclic cases,
unless MSBMM can be verified in $n^{\omega+o(1)}$ time.
\end{abstract}
\clearpage{}

\global\long\def\lifts{\mathit{LiftSet}}
\global\long\def\liftp{\mathit{Lift}}

\section{Introduction}

In the context of model checking, the simulation preorder of a transition
system is an abstraction that allows to reduce the state space, while
preserving the satisfiability of a large class of temporal logic formulas~\cite{Bustan2003}.
On Kripke structures (i.e.\ directed vertex-labeled state-transitions
graphs), it can be defined co-inductively: a state $t$ \emph{simulates}
a state $s$ whenever $t$ and $s$ are labeled in the same way and,
for every transition from $s$ to $s'$, there is a transition from
$t$ to $t'$, such that $t'$ simulates $s'$.

Being a crucial problem in model checking, the computation of the
simulation preorder has been studied extensively, both for explicitly
defined systems and for implicit transition systems arising from process
algebras. However, in the most basic setting of finite systems provided
explicitly, some fundamental complexity questions about this problem
are still open. In this work we address some of these questions, discovering
a close relationship, in terms of complexity issues, between this
problem and some variants of matrix multiplication.

\paragraph{Motivation.}

Transition systems are essentially labeled directed graphs, possibly
infinite, whose vertices represent states, edges represent the possible
transitions between states, and labels represent visible properties
of a system, such as~I/O. When reasoning about a transition system,
it is possible to ignore irrelevant details by means of \emph{abstractions},
which reduce the system to a smaller structure while preserving the
properties of the system that are being studied~\cite{Bustan2003}.
Abstractions are commonly expressed as equivalence/preorder relations
between the system states: when two states are considered equivalent
according to the abstraction, they can be collapsed into one. Several
abstractions have been defined in the literature~\cite{Baier2008},
the most important being bisimulation equivalence~\cite{Park1981},
simulation preorder~\cite{Milner1971} and trace equivalence~\cite{Hoare1978},
along with their respective variants. These abstractions, listed from
the finer to the coarser, preserve the validity of formulas in progressively
smaller fragments of the $\mu$-calculus logic~\cite{Bustan2003}.

 In this work we consider finite transition systems, with $n$ states
and $m\geq n$ transitions, whose graph is given explicitly and is
accompanied by labels on the states. (These vertex-labeled graphs
are called Kripke structures~\cite{Kripke1963}.) The bisimulation
equivalence, the finest among the mentioned abstractions, is computed
with an almost-optimal running time $O(m\log n)$ by the algorithm
of \cite{Paige1987}. At the opposite end, the computation of the
trace equivalence relation is known to be PSPACE-complete~\cite{Stockmeyer1973}.
The simulation preorder lies between these two extremes: the problem
is tractable, but no $o(n^{3})$-time algorithm is known (while only
$\Theta(n^{2})$ is needed for output).

Polynomial algorithms to compute the simulation preorder have been
first presented in \cite{Bloom1989,Cleaveland1993,Cleaveland1993a},
improved to $O(mn)$ time independently in \cite{Henzinger1995}
and \cite{Bloom1995}. These algorithms perform a fixpoint computation,
starting with the full relation, and then repeatedly removing pairs
of states as soon as they are discovered not to be in the simulation
relation. More recently, a large family of new algorithms has been
proposed \cite{Gentilini2002,Gentilini2003,Ranzato2007,Ranzato2010,Markovski2011,Cece2013}
whose running time depends not only on the size of the input system,
but also on the number $n^{*}$ of equivalence classes in the simulation
preorder relation. They work by keeping a partition of states into
blocks of possibly equivalent states: working mostly at the block
level, these algorithms run faster when $n^{*}$ is much smaller than
$n$. However, they still require $\Omega(n^{3})$ time in the worst
case, since the preorder relation could turn out to be a (non-trivial)
partial order, where all the blocks eventually reduce to singletons.

A question arises naturally: is it possible to obtain a subcubic algorithm
for the simulation preorder? Our work stems from the realization that
this problem hides a boolean matrix multiplication inside its belly,
and this explains why getting below the $\Omega(n^{3})$ time barrier
has been so difficult. Motivated by this result, we address the simulation
preorder from viewpoint of pure computational complexity, discovering
that the relationship with matrix multiplication is rich and many-sided.
Moreover, we study the existence of explicit certificates for the
simulation preorder, and the possibility to check the result more
efficiently than computing it from scratch. Despite being crucial
for a deep understanding of the algorithmic problem, to the best of
our knowledge, the analysis of certificates is lacking in previous
work.

\paragraph{Simulation as ``two-tokens'' games.}

To present our results, we first reduce the computation of the simulation
preorder to its essential underlying algorithmic problem, expressed
in terms of two-players \emph{reachability games}~\cite{DeAlfaro1998,Alur2002,Chatterjee2014}.
While the correspondence between simulation preorder and two-players
games is well-established in the literature~(see, e.g., \cite{Henzinger1995,Etessami2005,Cerny2010}),
in our analysis we point out the specific structure of simulation
games, with respect to general reachability games, and how this structure
can (or cannot) be exploited algorithmically. This is very relevant
for the problem: with a scrupulous eye, one can notice that the best-so-far
$O(nm)$-time algorithm for simulation does not exploit this structure
at all; actually, a known linear-time algorithm \cite{Beeri1980,Alur2002}
for reachability games, when applied to simulation games, achieves
the same running time~(see Remark~\ref{rem:nm-algorithm}). To achieve
any improvement along this line, the peculiarities of simulation games
need to be taken into account.

In a reachability game, the goal of the first player Alice is to reach
a configuration among a given set, while the second player Bob tries
to avoid this. If Alice manages to reach the goal, she wins, otherwise
the game continues forever and the victory is assigned to Bob. A \emph{simulation
game} is a particular reachability game, defined in terms of a given
Kripke structure. A configuration consists of a pair of states $(s,t)$.
Alice reaches her goal, and wins immediately, when $s$ and $t$ hold
different labels. Otherwise, she first chooses a transition $(s,s')$
from $s$, then Bob chooses a transition $(t,t')$ from $t$, and
the game moves to the next configuration $(s',t')$. It is well-known
that $t$ simulates $s$ iff Alice does not have a winning strategy
from $(s,t)$; in fact, the latter is sometimes used as the definition
of simulation preorder~\cite{Henzinger1995,Etessami2005}.

In this paper, we define another type of reachability games, called
\emph{two-tokens reachability games} (\noun{2TRG}), which generalize
simulation games. In a 2TRG there are two tokens, which are moved
in turn by the two players. Differently from simulation games, in
2TRGs the two tokens move along two distinct graphs. Moreover, in
2TRGs the set of goal configurations is arbitrary, and may include
configurations where either player holds the turn, not necessarily
Alice. Despite being more general than simulation games, we prove
that 2TRGs are not computationally harder to solve, hence strictly
equivalent. We analyze 2TRGs instead of simulation games, with a two-fold
advantage. First, we drop some of the assumptions which are not helpful
in studying the complexity of simulation games, thus reducing ourselves
to a more essential algorithmic problem. Second, we introduce a symmetry
between the roles of the two players which, thanks to dualization,
halves the work required to describe some of our results. The reduction
from 2TRGs to simulation games is quite simple, so the reader is left
with the opportunity to map our results directly to simulation games
with ease.

The solution of reachability games is easily characterized in terms
of \emph{closed sets} and \emph{progress measures} (concepts similar
to those used in the more involved Büchi games and parity games~\cite{Jurdzinski2000,Alur2002}),
which in turn can be defined naturally as (pre-/post-)fixpoints of
some \emph{lifting operators}~(see~\cite{Jurdzinski2000,Etessami2005,Chatterjee2014}).
When these notions are applied to 2TRGs, the relationship with matrix
multiplications appears clear: the lifting operators themselves are
expressible as forms of matrix multiplications.

\paragraph{Summary of contributions.}

As mentioned before, our first result is to show that computing the
simulation preorder on an $n$-state Kripke structure is at least
as hard as $n\times n$ boolean matrix multiplication (BMM), explaining
why obtaining a ``combinatorial'' subcubic algorithm seems to be
hard~\cite{Abboud2014}. Next, in the case of \emph{acyclic} Kripke
structures, we show that fast BMM can be employed to obtain a truly
subcubic algorithm, based on a divide-and-conquer technique, and running
in $n^{\omega+o(1)}$-time (where $\omega<2.376$ is the exponent
of matrix multiplication \cite{Coppersmith1990}). Since the previous
lower bound also applies to this restricted case, the simulation problem
on acyclic structures is essentially equivalent to BMM. Finally, for
the acyclic case, we exhibit $O(n^{2})$-size canonical certificates,
that can be checked via BMM verification. By transforming BMM into
a standard $(+,\times)$-matrix multiplication, this yields $O(n^{2})$-size
certificates that can be verified in $O(n^{2})$~time~\cite{Freivalds1977,Kimbrel1993,Korec2014}.

For cyclic structures, we also provide $O(n^{2})$-size canonical
certificates checkable in sub-cubic time; in this case, however, we
need to employ a more general variant of matrix multiplication. We
introduce the $\max$-semi-boolean matrix multiplication (MSBMM),
a multiplication on the semi-ring $(\max,\times)$, where one of the
two matrices contains only zeros and ones. This variant of matrix
multiplication is more general than BMM (consider the case where both
matrices contain only zeros and ones), but still admits a truly subcubic
solution. Indeed, it can be easily transformed into a $(\max,\min)$-product
(where one matrix contains only $+\infty$ and $-\infty$), which
in turn can be solved in $n^{(3+\omega)/2+o(1)}\leq O(n^{2.792})$
time \cite{Duan2009,Vassilevska2009}.

Our last contribution is a reduction from the verification of $n\times n$
MSBMM to the simulation preorder in a $O(n\log n)$-states cyclic
Kripke structure. This is by far the most involved reduction given
in this work, and relies on a construction of permutation networks
given by Waksman \cite{Waksman1968} in the '60s (which has applications
in the quite distant fields of telecommunications~\cite{Clos1953,Benes1965,Dally2004}
and parallel architectures~\cite{Culler1998}). This result implies
a separation between the acyclic and the cyclic case: our $n^{\omega+o(1)}$-time
lower bound cannot be matched in the acyclic case, unless MSBMM can
be also verified in $n^{\omega+o(1)}$ time. (By analogy with the
results in~\cite{Williams2010}, we do not expect the verification
version to be substantially easier than the computation.) Determining
whether this separation holds, i.e., whether $\max$-semi-boolean
multiplication is actually harder than boolean multiplication, remains
as an open question, sitting among the numerous other problems regarding
the complexity of matrix multiplications.

\paragraph{Paper organization.}

The rest of this paper is organized as follows. In Section~\ref{sec:reachability-games},
we define reachability games and we introduce some notions and classical
results about these games, which are useful for this work. In Section~\ref{sec:simulation-games},
we introduce two-tokens reachability games and simulation games, establishing
the equivalence between the two (Theorem~\ref{thm:2trg-sim-equiv}).
In Section~\ref{sec:acyclic}, we present our results for the acyclic
case, namely, the reduction from BMM to 2TRGs (Theorem~\ref{thm:bmm-to-sim}),
the certificates for 2TRGs with $O(n^{2})$ time verification (Theorem~\ref{thm:verify-to-bmm}),
and our subcubic divide-and-conquer algorithm (Theorem~\ref{thm:sim-to-bmm}).
In Section~\ref{sec:cyclic}, we present our results for the cyclic
case, namely, the certificates verifiable via $\max$-semi-boolean
matrix multiplications (Theorem~\ref{thm:verify-to-msbmm}) and the
reduction from MSBMM verification to 2TRGs (Theorem~\ref{thm:msbmm-to-sim}).

\section{Reachability games}

\label{sec:reachability-games}
\begin{defn}[Reachability games]
A \emph{game-graph} is a structure ${\cal G}=({\cal V},{\cal E},{\cal V}^{0},{\cal V}^{1})$
consisting of a finite set ${\cal V}$ of \emph{configurations}, a
set ${\cal E}\subseteq{\cal V}\times{\cal V}$ of \emph{moves}, and
a partition $({\cal V}^{0},{\cal V}^{1})$ of ${\cal V}$ into configurations
\emph{controlled} respectively by player~0 (Alice) and player~1
(Bob). Configurations and moves form a directed graph $({\cal V},{\cal E})$,
called \emph{configuration graph}.

A \emph{play} is a finite or infinite walk in the configuration graph
$({\cal V},{\cal E})$, i.e., a non-empty sequence of configurations
$\pi=\sigma_{0}\sigma_{1}\cdots\in{\cal V}^{+}\cup{\cal V}^{\infty}$,
such that $(\sigma_{i},\sigma_{i+1})\in{\cal E}$ is a move for every
two consecutive configurations $\sigma_{i}$ and $\sigma_{i+1}$ in
$\pi$.

A (positional\footnote{Non-positional strategies are not needed in this paper. From now on,
the term ``positional'' is omitted.})\emph{ strategy} for player $P\in\{0,1\}$ is a function $s\colon{\cal V}^{P}\to{\cal V}\cup\{\bot\}$,
such that, for every configuration $\sigma\in{\cal V}^{P}$ controlled
by~$P$, either $(\sigma,s(\sigma))\in{\cal E}$ is a move, or $s(\sigma)=\bot$.
Player $P$ \emph{moves} from $\sigma$ to $s(\sigma)$ if $(\sigma,s(\sigma))\in{\cal E}$,
and \emph{stops} on $\sigma$ if $s(\sigma)=\bot$. A play $\pi=\sigma_{0}\sigma_{1}\cdots$
is \emph{conforming to} $s$ if, for every configuration $\sigma_{i}$
in $\pi$, if $\sigma_{i}\in{\cal V}^{P}$ is controlled by $P$,
then either $(\sigma_{i},s(\sigma_{i}))\in{\cal E}$ and $\sigma_{i}$
is followed by $\sigma_{i+1}=s(\sigma_{i})$, or $s(\sigma_{i})=\bot$
and $\sigma_{i}$ is the last configuration of $\pi$ (i.e.\ $\pi=\sigma_{0}\cdots\sigma_{i}\in{\cal V}^{i+1}$).

A\emph{ reachability game} is a pair $({\cal G},{\cal F})$, consisting
of a game-graph ${\cal G}$ and a partition ${\cal F}=({\cal F}^{0},{\cal F}^{1})$
of ${\cal V}$ into \emph{winning final configurations} for player~$0$
and player~$1$, respectively. A play $\pi$ is \emph{winning} for
player~$P$ if it is finite (say, $\pi=\sigma_{0}\cdots\sigma_{\ell}\in{\cal V}^{\ell+1}$)
and its final configuration is winning\footnote{In the definition of reachability games usually found in the literature,
it is sufficient that a configuration in a certain set $F$ is reached,
anywhere in a play, and the play is considered winning for player~0.
In our definition, a play must be explicitly \emph{stopped} on a configuration
in ${\cal F}^{P}$, in order to be winning for player~$P$. In particular,
if the play reaches a configuration $\sigma\in{\cal F}^{P}\cap{\cal V}^{P}$,
then player~$P$ can stop at $\sigma$ and win immediately, but if
$\sigma\in{\cal F}^{P}\cap{\cal V}^{1-P}$, then player~$1-P$ may
choose a next move (if there is any) and the game continues. This
difference makes the description of our results simpler; nevertheless,
it is easy to establish an equivalence between the two variants.} for $P$ (i.e., $\sigma_{\ell}\in{\cal F}^{P}$). A play is \emph{surviving}
for player $P$ if it is either winning for~$P$ or infinite. I.e.,
infinite plays are neither winning nor losing for any player, but
they are surviving for both.

A strategy $s$ for player $P$ is \emph{winning} (resp.\ \emph{surviving})
from $\sigma_{0}\subseteq{\cal V}$, if every play $\sigma_{0}\sigma_{1}\cdots$
conforming to $s$ is winning (resp.\ surviving) for~$P$. The \emph{winning
set} ${\cal W}^{P}\subseteq{\cal V}$ (resp.\ the \emph{surviving
set} ${\cal S}^{P}\subseteq{\cal V}$) of $P$ is the set of configurations
from which $P$ has a winning (resp.\ surviving) strategy.
\end{defn}

\begin{defn}[Set lifting operator, closed set]
Let ${\cal U}\subseteq{\cal V}$. The set $\lifts^{P}({\cal U})\subseteq{\cal V}$
contains all the configurations from which player~$P$ can be sure
to either win immediately, or that at the next turn the game will
move to a configuration in ${\cal U}$. It is defined as follows:
\[
\sigma\in\lifts^{P}({\cal U})\iff\begin{cases}
\sigma\in{\cal F}^{P}\lor\bigvee_{(\sigma,\sigma')\in{\cal E}}\sigma'\in{\cal U} & \text{ if \ensuremath{\sigma\in{\cal V}^{P}}}\\
\sigma\in{\cal F}^{P}\land\bigwedge_{(\sigma,\sigma')\in{\cal E}}\sigma'\in{\cal U} & \text{ if \ensuremath{\sigma\in{\cal V}^{1-P}}}
\end{cases}
\]
where as usual $\bigvee_{x\in\emptyset}\coloneqq\mathrm{false}$ and
$\bigwedge_{x\in\emptyset}\coloneqq\mathrm{true}$. A set ${\cal U}\subseteq{\cal V}$
is \emph{closed} for $P$ if ${\cal U}\subseteq\lifts^{P}({\cal U})$.
\end{defn}

\begin{defn}[Potential, lifting operator, progress measure]
A \emph{potential} is a function $p\colon{\cal V}\to\mathbb{N}\cup\{\infty\}$.
Let $\supp(p)=\{\sigma\in{\cal V}\mid p(\sigma)<\infty\}$. We define
the potential $\liftp^{P}(p)$ as follows

\[
\liftp^{P}(p)(\sigma)=\begin{cases}
\mymin{}\:\{p_{\bot}^{P}(\sigma)\}\cup\{1+p(\sigma')\mid(\sigma,\sigma')\in{\cal E}\} & \text{ if \ensuremath{\sigma\in{\cal V}^{P}}}\\
\mymax{}\:\{p_{\bot}^{P}(\sigma)\}\cup\{1+p(\sigma')\mid(\sigma,\sigma')\in{\cal E}\} & \text{ if \ensuremath{\sigma\in{\cal V}^{1-P}}}
\end{cases}
\]
where $p_{\bot}^{P}(\sigma)=0$ if $\sigma\in{\cal F}^{P}$, $p_{\bot}^{P}(\sigma)=\infty$
if $\sigma\in{\cal F}^{1-P}$, and $1+\infty\coloneqq\infty$. A potential
$p$ is a \emph{progress measure} for $P$ if $p(\sigma)\geq\liftp^{P}(p)(\sigma)$
for every $\sigma\in{\cal V}$.\end{defn}
\begin{prop}[Characterizations of reachability games]
\label{prop:determinacy}The following properties hold:
\begin{enumerate}[label=(\alph*)]
\item \label{enu:stable}if ${\cal U}\subseteq{\cal V}$ is closed for $P$
then there is a strategy for $P$ surviving from every $\sigma\in{\cal U}$,
\item \label{enu:progress-measure}if $p$ is a progress measure for $P$
then there is a strategy for $P$ winning from every $\sigma\in\supp(p)$,
\item \label{enu:set-fixpoint}${\cal W}^{P}$ and ${\cal S}^{P}$ are respectively
the least and the greatest fixpoints of~$\lifts^{P}$,
\item \label{enu:potential-fixpoint}the operator $\liftp^{P}$ has a unique
fixpoint $r^{P}$ and we have ${\cal W}^{P}=\supp(r^{P})$ and ${\cal S}^{1-P}={\cal V}\setminus\supp(r^{P})$,
\item \label{enu:algorithm}this fixpoint $r^{P}$ can be computed in linear
time $O(|{\cal V}|+|{\cal E}|)$.
\end{enumerate}
\end{prop}
\begin{proof}
Properties \ref{enu:stable}--\ref{enu:potential-fixpoint} are variants
of classical results in infinite two-players games~\cite{Jurdzinski2000,Etessami2005,Chatterjee2014},
and applications of the Tarski theorem. As for property~\ref{enu:algorithm},
the measure $r^{P}$ can be computed by an alternating backward search~\cite{Alur2002,Beeri1980}.
A rigorous proof of all the properties is given in Appendix~\ref{sec:apx-games}.
\end{proof}

\section{Two-tokens and simulation games}

\label{sec:simulation-games}
\begin{defn}[Two-tokens reachability games]
A \emph{two-tokens game-graph} over two finite directed graphs $G_{0}=(V_{0},E_{0})$
and $G_{1}=(V_{1},E_{1})$, is a game-graph ${\cal G}(G_{0},G_{1})$
defined as follows. For every player $P\in\{0,1\}$, and every pair
of vertices $u\in V_{P}$ and $v\in V_{1-P}$, there is a configuration
$\langle P,u,v\rangle$, controlled by $P$. In the configuration
$\langle P,u,v\rangle$, a token belonging to $P$ is located on $u\in V_{P}$,
and a token belonging to $1-P$ is located on $v\in V_{1-P}$. The
player holding the turn can move her token along any edge $(u,u')\in E_{P}$
of her graph $G_{P}$, and then pass the turn to the other player,
resulting in the move $(\langle P,u,v\rangle,\langle1-P,v,u'\rangle)$.
(Observe that, to keep symmetry between the two players, $u'$ and
$v$ are swapped.) Summarizing, ${\cal G}(G_{0},G_{1})={\cal G}=({\cal V},{\cal E},{\cal V}^{0},{\cal V}^{1})$,
where ${\cal V}^{P}=\{\langle P,u,v\rangle\mid u\in V_{P},v\in V_{1-P}\}$,
${\cal V}={\cal V}^{0}\cup{\cal V}^{1}$, and ${\cal E}=\{(\langle P,u,v\rangle,\langle1-P,v,u'\rangle)\mid P\in\{0,1\},(u,u')\in E_{P},v\in V_{1-P}\}$.

A \emph{two-tokens reachability game} (2TRG) is a reachability game
$({\cal G},{\cal F})$ with ${\cal G}={\cal G}(G_{0},G_{1})$. The
problem \noun{2TRG Winning Set} (\noun{2TRG-WS}) of \emph{order} $|V_{0}|\times|V_{1}|$
asks to compute its winning and survival sets, given the graphs $G_{0}$,
$G_{1}$ and the partition ${\cal F}$. In the \noun{Acyclic} variant,
$G_{0}$ and $G_{1}$ are acyclic. In the \noun{Semi-Acyclic} variant,
at least one among $G_{0}$ and $G_{1}$ is acyclic.
\end{defn}

\begin{defn}[Kripke structure]
A \emph{Kripke structure} is a structure ${\cal K}=(S,T,L)$ consisting
of a set of states $S$, a transition relation $T\subseteq S\times S$,
and a labeling function $L\colon S\to\Lambda$ over the states.\footnote{The transition relation $T$ is sometimes required to be left-total
(i.e.\ $\forall s\in S\,\exists s'\in S$ such that $(s,s')\in T$),
and the label universe $\Lambda$ is usually defined as the power
set of a given set of atomic proposition. These requirements are not
relevant to our discussion, and have been omitted. According to the
definition above, a Kripke structure is nothing more than a vertex-labeled
directed graph.}
\end{defn}

\begin{defn}[Simulation game]
\label{def:simulation-game}A \emph{simulation game} for a Kripke
structure ${\cal K}=(S,T,L)$ is a 2TRG $({\cal G},{\cal F})$, where
the two-tokens game-graph ${\cal G}={\cal G}(G,G)$ is built over
two copies of the graph $G=(S,T)$, and ${\cal F}^{1}=\{\langle0,s,t\rangle\in{\cal V}^{0}\mid L(s)=L(t)\}$.
A state $t\in S$ \emph{simulates} a state $s\in S$ (written $s\preceq_{s}t$)
if $\langle0,s,t\rangle\in{\cal S}^{1}$. The relation $\preceq_{s}$
is called\footnote{The equivalence between this definition of simulation preorder and
a more classical one is given in Appendix~\ref{sec:apx-simulation-classical}.} \emph{simulation preorder}.

Computing the relation $\preceq_{s}$ over $S\times S$ is an instance
of the problem \noun{Simulation} of \emph{order} $|S|$. In the \noun{Acyclic}
variant of \noun{Simulation}, the graph $(S,T)$ is required to be
acyclic.\end{defn}
\begin{rem}
\label{rem:nm-algorithm}The simulation preorder can be computed in
$O(nm)$ time (assuming $m\geq n$) by solving $({\cal G},{\cal F})$
as in Proposition~\ref{prop:determinacy}~\ref{enu:algorithm},
since $|{\cal V}|=O(n^{2})$, $|{\cal E}|=O(nm)$, and ${\cal G}$
can be constructed efficiently. We point out that the classical $O(mn)$
time algorithms for simulation~\\cite{Henzinger1995,Bloom1995} can
be regarded as more elaborated instantiations of this algorithm.\end{rem}
\begin{thm}[2TRGs and simulation games are equivalent]
\label{thm:2trg-sim-equiv}Given any (acyclic) 2TRG $({\cal G},{\cal F})$
of order $n_{0}\times n_{1}$, there exist an (acyclic) simulation
game $({\cal G}',{\cal F}')$, on a Kripke structure ${\cal K}=(S,T,L)$,
and a map\footnote{We use subscripts to distinguish between objects associated with $({\cal G},{\cal F})$
and with $({\cal G}',{\cal F}')$.} $f\colon{\cal V}_{{\cal G}}\to{\cal V}_{{\cal G}'}$, such that $|S|=n=O(n_{0}+n_{1})$,
the structure ${\cal K}$ is constructible in $O(n^{2})$ time, $f$
is computable in $O(1)$ time, and $\langle P,u,v\rangle\in{\cal S}_{{\cal G},{\cal F}}^{1}\iff f(\langle P,u,v\rangle)\in{\cal S}_{{\cal G}',{\cal F}'}^{1}$.\end{thm}
\begin{proof}
Let $S=V_{0}\cup V_{0}^{*}\cup V_{1}$ where $V_{0}^{*}=\{u^{*}\mid u\in V_{0}\}$
(assuming unions are disjoint). Label each state $u^{*}\in V_{0}^{*}$
with a distinct label $L(u^{*})=\lambda_{u}$, and all the other states
$x\in V_{0}\cup V_{1}$ with the same label $L(x)=\lambda_{\square}$.
Let $T=E_{0}\cup E_{1}\cup\{(u,u^{*})\mid u\in V_{0}\}\cup\{(v,u)\mid\langle1,v,u\rangle\in{\cal F}^{1}\}\cup\{(v,u^{*})\mid\langle0,u,v\rangle\in{\cal F}^{1}\}$,
and let $f\colon{\cal V}_{{\cal G}}\to{\cal V}_{{\cal G}'}$ be the
inclusion.

($\implies$) Bob survives in $({\cal G}',{\cal F}')$ with the same
strategy as in $({\cal G},{\cal F})$, unless one of the following
occurs: (a) the strategy of Bob says to stop and win on a configuration
$\langle1,v,u\rangle\in{\cal F}^{1}$, or (b) Alice moves to a starred
node $u^{*}$, say, from $\langle0,u,v\rangle$ to $\langle1,v,u^{*}\rangle$.
In case (a), instead of stopping, take the edge $(v,u)$ given by
construction. In case (b) take the edge $(v,u^{*})$, which is present
since $\langle0,u,v\rangle\in{\cal F}^{1}$ (otherwise Alice would
have won in $({\cal G},{\cal F})$ by stopping on $\langle0,u,v\rangle$).
In both cases~(a) and (b), Bob moves to a configuration $\langle0,x,x\rangle$
where the two tokens are located on the same vertex. From now on,
Bob can copy all the moves of Alice and survive forever.

($\impliedby$) Suppose that Alice wins in $({\cal G},{\cal F})$
from $\langle P,u,v\rangle$. She survives in $({\cal G}',{\cal F}')$
applying the same strategy, until one of the following occurs: (a)
the strategy says to stop on some winning final configuration $\langle0,u,v\rangle\in{\cal F}^{0}$,
or (b) Bob moves from $\langle1,v,u\rangle$ to $\langle0,u,x\rangle$
with $x\in V_{0}\cup V_{0}^{*}$. In both cases, take the edge $(u,u^{*})$.
We show that Bob cannot move to $u^{*}$ in the next turn, and since
$u^{*}$ is the only state labeled with $\lambda_{u}$, Alice wins.
In case (a), there is no edge $(v,u^{*})$ since $\langle0,u,v\rangle\in{\cal F}^{0}$.
In case (b), we have $\langle1,v,u\rangle\in{\cal F}^{0}$ (otherwise
Bob could have won in $({\cal G},{\cal F})$ by stopping on $\langle1,v,u\rangle$),
so $x\neq u$ and there exists no edge $(x,u^{*})$. (See Appendix~\ref{sec:apx-2trg-sim-equiv}
for a more formal proof.)
\end{proof}

\section{Acyclic case}

\label{sec:acyclic}

We start by showing that 2TRGs are at least as hard as boolean matrix
multiplication.
\begin{defn}[Boolean matrix multiplication]
Given an $n_{1}\times n_{2}$ boolean matrix $\mathbf{B}_{1}$ and
an $n_{2}\times n_{3}$ boolean matrix $\mathbf{B}_{2}$, their boolean
product is the $n_{1}\times n_{3}$ boolean matrix $\mathbf{B}_{1}\star\mathbf{B}_{2}$
defined by: $(\mathbf{B}_{1}\star\mathbf{B}_{2})[i,j]=\bigvee_{k=1}^{n_{2}}\mathbf{B}_{1}[i,k]\land\mathbf{B}_{2}[k,j]$.
The problem \noun{Boolean Matrix Multiplication} (BMM) of size $n_{1}\times n_{2}\times n_{3}$
asks to compute $\mathbf{B}_{1}\star\mathbf{B}_{2}$ given $\mathbf{B}_{1}$
and $\mathbf{B}_{2}$.\end{defn}
\begin{thm}
\label{thm:bmm-to-sim}If \noun{2TRG Winning Set} of order $n\times n$
can be solved in $O(n^{\alpha})$ time, for some $\alpha\geq2$, then
\noun{Boolean Matrix Multiplication} of size $n\times n\times n$
can be computed in $O(n^{\alpha})$ time.\end{thm}
\begin{proof}
We want to compute the boolean product $\mathbf{B}_{1}\star\mathbf{B}_{2}$
between two $n\times n$ matrices. Consider a \noun{2TRG} where $V_{0}=\{x_{1},\dots,x_{n}\}\cup\{y_{1},\dots,y_{n}\}$
and $V_{1}=\{z_{1},\dots,z_{n}\}$. The edges of $G_{0}$ go only
from nodes of type $x$ to nodes of type $y$, and are defined using
$\mathbf{B}_{1}$, namely $E_{0}=\{(x_{i},y_{k})\mid\mathbf{B}_{1}[i,k]=1\}$.
The graph $G_{1}$ has no edges. The matrix $\mathbf{B}_{2}$ is used
to define the goal configurations for Alice ${\cal F}^{0}=\{\langle1,z_{j},y_{k}\rangle\mid\mathbf{B}_{2}[k,j]=1\}$.

Suppose the play starts from $\langle0,x_{i},z_{j}\rangle$. Since
Bob has no moves, the play lasts at most one turn. Hence, the only
way for Alice to win is to reach in a single move a winning final
configuration $\langle1,z_{j},y_{k}\rangle\in{\cal F}^{0}$. This
is possible iff, for some $k\in\{1,\dots,n\}$, we have both $(x_{i},y_{k})\in E_{0}$
and $\langle1,z_{j},y_{k}\rangle\in{\cal F}^{0}$, i.e.\ both $\mathbf{B}_{1}[i,k]=1$
and $\mathbf{B}_{2}[k,j]=1$. That is, we have $\langle0,x_{i},z_{j}\rangle\in{\cal W}^{0}\iff(\mathbf{B}_{1}\star\mathbf{B}_{2})[i,j]=1$,
so by computing the winning set ${\cal W}^{0}$ we compute the product
$\mathbf{B}_{1}\star\mathbf{B}_{2}$.
\end{proof}
The following lemma is used both for verification (Theorem~\ref{thm:verify-to-bmm})
and, later, for our divide-and-conquer algorithm.
\begin{lem}[Computing $\lifts$ via BMM]
\label{lem:f-to-bmm}The operator $\lifts^{P}$ in \noun{2TRG}s of
order at most $n\times n$ can be computed (verified) by computing
(verifying) two boolean matrix multiplications of size at most $n\times n\times n$,
with only $O(n^{2})$ extra time.\end{lem}
\begin{proof}
Write $V_{P}=\{u_{1},\dots,u_{n_{P}}\}$ and $V_{1-P}=\{v_{1},\dots,v_{n_{1-P}}\}$.
Given a configuration $\langle P,u_{i},v_{j}\rangle\in{\cal V}^{P}$,
where player $P$ holds the turn, we have by definition $\langle P,u_{i},v_{j}\rangle\in\lifts^{P}({\cal U})$
iff she can either win immediately ($\langle P,u_{i},v_{j}\rangle\in{\cal F}^{P}$)
or she can move to some configuration $\langle1-P,v_{j},u_{k}\rangle\in{\cal U}$.
This holds iff, for some $k\in\{1,\dots,n_{P}\}$, we have both $(u_{i},u_{k})\in E_{P}$
and $(1-P,v_{j},u_{k})\in{\cal U}$. By considering the adjacency
matrix $\mathbf{E}$ of $G_{P}$ (i.e., a $n_{P}\times n_{P}$ boolean
matrix with $\mathbf{E}[i,j]=1$ iff $(u_{i},u_{j})\in E_{P}$), and
the $n_{P}\times n_{1-P}$ boolean matrix $\mathbf{U}$ where $\mathbf{U}[i,j]=1$
iff $(1-P,v_{j},u_{i})\in{\cal U}$, this is equivalent to saying
that $(\mathbf{E}\star\mathbf{U})[i,j]=1$. In formulas: 
\[
\langle P,u_{i},v_{j}\rangle\in\lifts^{P}({\cal U})\iff\langle P,u_{i},v_{j}\rangle\in{\cal F}^{P}\lor(\mathbf{E}\star\mathbf{U})[i,j].
\]

For those configurations $\langle1-P,v_{j},u_{i}\rangle\in{\cal V}^{1-P}$,
where player $P$ does not hold the turn, we can work by dualization.
Indeed, by applying the de~Morgan law $\langle1-P,v_{j},u_{i}\rangle\in\lifts^{P}({\cal U})\iff\langle1-P,v_{j},u_{i}\rangle\notin\lifts^{1-P}({\cal V}\setminus{\cal U})$,
we reduce ourselves to the previous case, where $P$ is substituted
with $1-P$ and ${\cal U}$ with ${\cal V}\setminus{\cal U}$. Summarizing,
one BMM is needed to compute $\lifts^{P}({\cal U})\cap{\cal V}^{P}$
and a second BMM is needed to compute $\lifts^{P}({\cal U})\cap{\cal V}^{1-P}$,
by dualization. It is clear that we spend no more than $O(n^{2})$
time, besides the computation of the two BMMs.\end{proof}
\begin{thm}
\label{thm:verify-to-bmm}\noun{(Semi-)Acyclic 2TRG-WS} of order $n\times n$
can be verified via boolean matrix product verification of size $n\times n\times n$,
with only $O(n^{2})$ extra time.

\noun{Acyclic Simulation} of order $n$ and \noun{(Semi-)Acyclic 2TRG-WS}
of order $n\times n$ admit a $O(n^{2})$-size certificate that can
be verified via standard $(+,\times)$-matrix product verification
of size $n\times n\times n$.\end{thm}
\begin{proof}
If any of $G_{0}$ and $G_{1}$ is acyclic, then ${\cal G}={\cal G}(G_{0},G_{1})$
is also acyclic and there are no infinite plays. Hence, the set ${\cal S}^{P}={\cal W}^{P}$
is the unique solution of the equation $\lifts^{P}({\cal U})={\cal U}$.
To verify that ${\cal U}={\cal S}^{P}$ for a given set ${\cal U}\subseteq{\cal V}$,
it is sufficient to verify that this equation holds, which, by Lemma~\ref{lem:f-to-bmm},
is equivalent to verifying two BMMs. The result of the standard $(+,\times)$-matrix
multiplications corresponding to these two BMMs can be used as a certificate,
so that they be verified in $O(n^{2})$ time.
\end{proof}
To describe our algorithm for acyclic 2TRGs (Theorem~\ref{thm:sim-to-bmm}),
we first present our approach on reachability games.
\begin{defn}[Induced sub-game-graph and sub-partition]
Let ${\cal G}=({\cal V},{\cal E},{\cal V}^{0},{\cal V}^{1})$ be
a game-graph and ${\cal U}\subseteq{\cal V}$. The \emph{sub-game-graph}
of ${\cal G}$ \emph{induced} by ${\cal U}$ is the game-graph ${\cal G}[{\cal U}]=({\cal U},{\cal E}\cap({\cal U}\times{\cal U}),{\cal V}^{0}\cap{\cal U},{\cal V}^{1}\cap{\cal U})$.
Given a partition ${\cal F}=({\cal F}^{0},{\cal F}^{1})$ of ${\cal V}$,
the \emph{sub-partition} ${\cal F}[{\cal U}]$ \emph{induced} by ${\cal U}$
is $({\cal F}^{0}\cap{\cal U},{\cal F}^{1}\cap{\cal U})$.\end{defn}
\begin{lem}[Dicut decomposition of reachability games]
\label{lem:dicut}Let $({\cal V}_{T},{\cal V}_{H})$ be a dicut of
the configuration graph $({\cal V},{\cal E})$, i.e., a partition
$({\cal V}_{T},{\cal V}_{H})$ of ${\cal V}$ such that there are
no edges from ${\cal V}_{H}$ to ${\cal V}_{T}$. Define the following:
\begin{enumerate}[label=(\arabic*)]
\item ${\cal S}_{H}^{P}$ is the surviving set of $P$ in the game $({\cal G}[{\cal V}_{H}],{\cal F}[{\cal V}_{H}])$,
\item ${\cal F}_{*}=({\cal F}_{*}^{0},{\cal F}_{*}^{1})$ is a partition
of ${\cal V}$ with ${\cal F}_{*}^{P}=\lifts^{P}({\cal S}_{H}^{P})$,
\item ${\cal S}_{T}^{P}$ is the surviving set of $P$ in the game $({\cal G}[{\cal V}_{T}],{\cal F}_{*}[{\cal V}_{T}])$.
\end{enumerate}
Then, the surviving set ${\cal S}^{P}$ in the original game can be
written as ${\cal S}^{P}={\cal S}_{H}^{P}\cup{\cal S}_{T}^{P}$.\end{lem}
\begin{proof}
Fix $P=0$. If a play starts with a configuration $\sigma_{0}\in{\cal V}_{H}$,
then it never reaches any configuration in ${\cal V}_{T}$, since
there are no backward edges in the dicut. Hence, for all the initial
configurations in ${\cal V}_{H}$, the problem is equivalent in the
sub-game $({\cal G}[{\cal V}_{H}],{\cal F}[{\cal V}_{H}])$, i.e.\ ${\cal S}^{0}\cap{\cal V}_{H}={\cal S}_{H}^{0}$.

Consider now a configuration $\sigma\in{\cal V}^{0}$ where Alice
holds the turn. If she can move to a configuration $\sigma'\in{\cal S}_{H}^{0}$,
which is surviving, then we can assume that she will take this opportunity
and survive. Hence, $\sigma$ can be added to the final winning configurations
of Alice, and all the outgoing moves from $\sigma$ can be removed.
Indeed, after this change, $\sigma$ will still be a surviving configuration
(actually, winning) for Alice. Now take a configuration $\sigma\in{\cal V}^{1}$
where Bob holds the turn. If he has no other choice but to either
stop, and lose immediately, or move to a configuration $\sigma'\in{\cal S}_{H}^{0}$,
surviving for his opponent Alice, then he clearly cannot win from
$\sigma$. Hence, also in this case, $\sigma$ can be added to the
final winning configurations of Alice and all the outgoing moves removed.

In general, the new set ${\cal F}_{*}^{0}$ of winning final configurations
for Alice can be defined as ${\cal F}_{*}^{0}=\lifts^{0}({\cal S}_{H}^{0})$,
and all the moves from ${\cal V}_{T}$ to ${\cal V}_{H}$ can be removed.
To solve the problem for the second part of the game-graph, we can
now work on the sub-game-graph ${\cal G}[{\cal V}_{T}]$, but only
after replacing the winning final configurations with the new partition
${\cal F}_{*}[{\cal V}_{T}]$. We obtain ${\cal S}^{0}\cap{\cal V}_{T}={\cal S}_{T}^{0}$
and the statement of the lemma follows. (See Appendix~\ref{sec:apx-dicut}
for a more formal proof.)\end{proof}
\begin{thm}
\noun{\label{thm:sim-to-bmm}Acyclic Simulation} and \noun{Acyclic
2TRG Winning Set} can be computed in $n^{\omega+o(1)}$ time, for
any $\omega$ such that boolean matrix multiplication can be solved
in $n^{\omega+o(1)}$ time.\end{thm}
\begin{proof}
Let $(V_{0}^{T},V_{0}^{H})$ be a dicut of $G_{0}$ (i.e., a partition
of $V$ such that $E\cap(V_{0}^{H}\times V_{0}^{T})=\emptyset$) with
$|V_{0}^{T}|,|V_{0}^{H}|\leq\lceil n/2\rceil$. Such a dicut can be
easily obtained from a topological sort of $G_{0}$, splitting at
about half. Observe that the dicut $(V_{0}^{T},V_{0}^{H})$ induces
a dicut $({\cal V}_{T},{\cal V}_{H})$ of the configuration graph
of ${\cal G}={\cal G}(G_{0},G_{1})$, with ${\cal G}[{\cal V}_{X}]={\cal G}(G_{0}[V_{0}^{X}],G_{1})$
for $X\in\{T,H\}$. To compute ${\cal S}^{0}$, we apply the formula
${\cal S}^{0}={\cal S}_{H}^{0}\cup{\cal S}_{T}^{0}$ given by Lemma~\ref{lem:dicut},
where ${\cal S}_{H}^{0}$ and ${\cal S}_{T}^{0}$ are computed recursively
and ${\cal F}_{*}^{0}=\lifts^{0}({\cal S}_{H}^{0})$ is computed via
fast BMM in $n^{\omega+o(1)}$ time (by Lemma~\ref{lem:f-to-bmm}).
Crucially, at each recursive call we dualize the game, swapping the
two players. The running time $T(n_{0},n_{1})$ then satisfies the
recurrence $T(n_{0},n_{1})\leq2T(n_{1},\lceil n_{0}/2\rceil)+(n_{0}+n_{1})^{\omega+o(1)}\leq4T(\lceil n_{0}/2\rceil,\lceil n_{1}/2\rceil)+(n_{0}+n_{1})^{\omega+o(1)}$.
Under the assumption $\omega\geq2$, we get $T(n,n)\leq n^{\omega+o(1)}$.
(If $\omega=2$, the extra logarithmic factor is accounted for in
the $n^{o(1)}$ term.)
\end{proof}

\section{Cyclic case}

\label{sec:cyclic}We first show how $\max$-semi-boolean matrix multiplication
can be employed for the verification of \noun{2TRG Winning Set} (Lemma~\ref{thm:verify-to-msbmm}).
\begin{defn}[$\min$-/$\max$-semi-boolean matrix multiplication]
 Given an $n_{1}\times n_{2}$ matrix of numbers\footnote{Integers, reals, or elements of any totally ordered set.}
$\mathbf{A}$ and an $n_{2}\times n_{3}$ boolean matrix $\mathbf{B}$,
their $\min$- and $\max$-semi-boolean products are the $n_{1}\times n_{3}$
matrices $\mathbf{A}\star_{\min}\mathbf{B}$ and $\mathbf{A}\star_{\max}\mathbf{B}$
defined as follows:
\begin{eqnarray*}
(\mathbf{A}\star_{\min}\mathbf{B})[i,j] & = & \,\min\;\{\mathbf{A}[i,k]\mid k=1,\dots,n_{2}\text{ and }\mathbf{B}[k,j]=1\}\\
(\mathbf{A}\star_{\max}\mathbf{B})[i,j] & = & \max\;\{\mathbf{A}[i,k]\mid k=1,\dots,n_{2}\text{ and }\mathbf{B}[k,j]=1\}.
\end{eqnarray*}
The problems \noun{Min-} and \noun{Max-Semi-Boolean Matrix Multiplication}
(MSBMM) of size $n_{1}\times n_{2}\times n_{3}$ ask to compute $\mathbf{A}\star_{\min}\mathbf{B}$
and $\mathbf{A}\star_{\max}\mathbf{B}$ given $\mathbf{A}$ and $\mathbf{B}$.
In the \noun{Distinct} variant of MSBMM, we require $\mathbf{A}[i,k]\neq\mathbf{A}[i,k']$
for $k\neq k'$.
\end{defn}
The min and max versions are clearly equivalent since $\mathbf{A}\star_{\min}\mathbf{B}=-((-\mathbf{A})\star_{\max}\mathbf{B})$.
Observe that, given a MSBMM, we can replace $\mathbf{A}[i,k]$ with
its rank in the set $\{\mathbf{A}[i,k]\mid k=1,\dots,n_{2}\}$, breaking
ties arbitrarily, and we get an equivalent \noun{Distinct }MSBMM problem.
\begin{lem}[Computing $\liftp$ via MSBMM]
\label{lem:liftp-msbmm}The operator $\liftp$ in \noun{2TRG}s of
order at most $n\times n$ can be computed (verified) by computing
(verifying) two \noun{MSBMM} of size at most $n\times n\times n$,
with only $O(n^{2})$ extra time.\end{lem}
\begin{proof}
Write $V_{P}=\{u_{1},\dots,u_{n_{0}}\}$ and $V_{1-P}=\{v_{1},\dots,v_{n_{1}}\}$.
Consider the adjacency matrix $\mathbf{E}$ of $G_{P}$ (an $n_{P}\times n_{P}$
boolean matrix with $\mathbf{E}_{P}[i,j]=1$ iff $(u_{i},u_{j})\in E_{P}$),
and define the $n_{P}\times n_{1-P}$ matrix $\mathbf{P}$ with $\mathbf{P}[i,j]=p(\langle1-P,v_{j},u_{i}\rangle)$.
For a configuration $\langle P,u_{i},v_{j}\rangle\in{\cal V}^{P}$,
where player~$P$ holds the turn, we have $\liftp^{P}(p)(\langle P,u_{i},v_{j}\rangle)=0$
if $\langle P,u_{i},v_{j}\rangle\in{\cal F}^{P}$, and otherwise 
\[
\liftp^{P}(p)(\langle P,u_{i},v_{j}\rangle)=\min_{(u_{i},u_{k})\in E_{P}}p(\langle1-P,v_{j},u_{k}\rangle)=\min\{\mathbf{P}[k,j]\mid\mathbf{E}[i,k]=1\}=(\mathbf{P}\star_{\min}\mathbf{E})[i,j].
\]
To compute $\liftp^{P}(p)$ for those configurations in which the
other player $1-P$ holds the turn, we dualize the problem: we show
equivalently how to compute $\liftp^{1-P}(p)(\langle P,u_{i},v_{j}\rangle)$
for $\langle P,u_{i},v_{j}\rangle\in{\cal V}^{P}$. Similarly as in
the previous case, we obtain $\liftp^{1-P}(p)(\langle P,u_{i},v_{j}\rangle)=\infty$
if $\langle P,u_{i},v_{j}\rangle\in{\cal F}^{1-P}$, and $\liftp^{1-P}(p)(\langle P,u_{i},v_{j}\rangle)=(\mathbf{P}\star_{\max}\mathbf{E})[i,j]$
otherwise.

Hence, a total of two MSBMMs are needed to compute the potential $\liftp^{P}(p)$:
one for the configurations in ${\cal V}^{P}$, and another for the
configurations in ${\cal V}^{1-P}$, after dualization.\end{proof}
\begin{thm}
\noun{\label{thm:verify-to-msbmm}Simulation} of order $n$ and \noun{2TRG
Winning Set} of order $n\times n$ admit $O(n^{2})$-size canonical
certificates that can be checked by verifying two \noun{Max-Semi-Boolean
Matrix Multiplications} of size $n\times n\times n$, and only $O(n^{2})$
extra time.\end{thm}
\begin{proof}
Recall that $r^{P}$ is the only solution of the equation $r^{P}=\liftp^{P}(r^{P})$,
and that ${\cal W}^{P}=\supp(r^{P})$. Hence, $r^{P}$ is a $O(n^{2})$-size
certificate, and it can be checked by verifying the fixpoint equation
using two MSBMMs (by Lemma~\ref{lem:liftp-msbmm}).
\end{proof}
The rest of this section is devoted to proving the following theorem.
\begin{thm}
\label{thm:msbmm-to-sim}The verification of \noun{Distinct Max-Semi-Boolean
Matrix Multiplication} of size $n\times m\times m$ can be reduced
to the verification of \noun{2TRG Winning Set} of order $n\times m\log m$.\end{thm}
\begin{cor}
If \noun{Simulation} of order $n$ can be computed or verified in
$O(n^{\alpha})$ time for $\alpha\geq2$, then \noun{Distinct MSBMM}
of size $n\times n\times n$ can be verified in $O(n^{\alpha}\log n)$
time.
\end{cor}
An $m\times m$ boolean matrix $\mathbf{B}$ and two $n\times m$
matrices of numbers $\mathbf{A}$ and $\mathbf{C}$ are given, where
$\mathbf{A}[i,k]\neq\mathbf{A}[i,k']$ for $k\neq k'$. We want to
check whether $\mathbf{C}[i,j]=(\mathbf{A}\star_{\max}\mathbf{B})[i,j]$
for every $i$ and $j$. Fixed $i\in\{1,\dots,n\}$ and $j\in\{1,\dots,m\}$,
let $k_{ij}$ be the only index such that $\mathbf{A}[i,k_{ij}]=\mathbf{C}[i,j]$.
If there is no such $k_{ij}$, or $\mathbf{B}[k_{ij},j]=0$, then
clearly the answer is no. Otherwise, $\mathbf{C}[i,j]\leq(\mathbf{A}\star_{\max}\mathbf{B})[i,j]$
for every $i,j$. It remains to check that there is no triple $(i,j,k)$
such that $\mathbf{A}[i,k]>\mathbf{C}[i,j]$ with $\mathbf{B}[k,j]=1$.
We call such a triple an \emph{invalid triangle}. We construct in
$O(nm\log m)$ time a 2TRG $({\cal G}(G_{0},G_{1}),{\cal F})$, where
Bob survives on some initial configurations iff there exists an invalid
triangle, and we conclude by checking that Alice wins from every configuration.

The graph $G_{0}$ contains $n$ isolated loops, i.e., $V_{0}=\{1,\dots,n\}$
and $E_{0}=\{(i,i)\mid i\in V_{0}\}$, so that the token of Alice
always remains in its initial position $i\in V_{0}$. The graph $G_{1}$
is built in such a way that, if (and only if) there is an invalid
triangle $(i,j,k)$, then Bob can move his token in a cycle, without
encountering losing configurations, and survive. Since the graph $G_{1}$
cannot depend on $i$, to achieve this goal we can only manipulate
the winning final configurations. Moreover, we need to keep the number
of vertices low to $O(m\log m)$. To this end, it comes to help the
construction of permutation networks.

A \emph{permutation network} \cite{Waksman1968} of size $n$ is
defined as follows. There are $n$ \emph{inlets} $u_{1},\dots,u_{n}$
and $n$ \emph{outlets} $v_{1},\dots,v_{n}$. Between the inlets and
the outlets, there is a set of \emph{gates} $S$, and for each gate
$s\in S$ there are two input ports $x_{1}^{s},x_{2}^{s}$ and two
output ports $y_{1}^{s},y_{2}^{s}$. A gate connects each of the two
input ports to an output port: when the gate gate is active, they
are crossed and swapped, otherwise they are connected in order. Inlets,
gate ports and outlets are connected with \emph{wires}, which form
a bijective relation $W\subseteq O\times I$ between $O=\{u_{1},\dots,u_{n}\}\cup\bigcup_{s\in S}\{y_{1}^{s},y_{2}^{s}\}$
and $I=\bigcup_{s\in S}\{x_{1}^{s},x_{2}^{s}\}\cup\{v_{1},\dots,v_{n}\}$.
The property of the network is as follows: for every permutation $\pi\colon\{1,\dots n\}\to\{1,\dots n\}$,
there is a subset of the gates $A_{\pi}\subseteq S$ which have to
be activated, so that the network realizes the permutation $\pi$
between the inlets and the outlets.\footnote{Given $\pi$, define the directed graph $G_{\pi}=(I\cup O,W\cup T_{\pi})$,
where $T_{\pi}\subseteq I\times O$ contains all the pairs of the
form $(x_{i}^{s},y_{j}^{s})$, for $s\in S$ and $i,j\in\{1,2\}$,
with $i\neq j$ if $s\in A_{\pi}$ and $i=j$ otherwise. For every
permutation $\pi$, the graph $G_{\pi}$ is the union of $n$ vertex-disjoint
paths $P_{1},\dots,P_{n}$, where $P_{i}$ goes from $u_{i}$ to $v_{\pi(i)}$.} Waksman \cite{Waksman1968} shows a construction of permutation networks
of size $n=2^{k}$ where $|S|=O(n\log n)$ and $A_{\pi}$ is computable
in $O(n\log n)$ time for every $\pi$.

\begin{figure}
\begin{centering}
\includegraphics{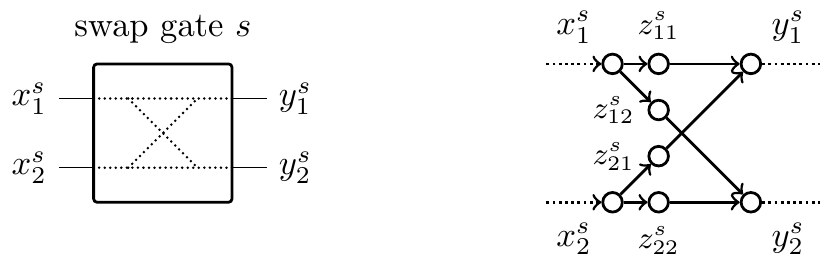}
\par\end{centering}

\centering{}\caption{\label{fig:gate}Visual representation of a swap gate and its corresponding
gate gadget graph.}
\end{figure}

To realize a permutation network in the graph $G_{1}$, we have to
show how to implement a gate. We define the \emph{gate gadget} graph
(Fig.~\ref{fig:gate}) as follows: there are two input vertices $x_{1}^{s}$
and $x_{2}^{s}$ and two output vertices $y_{1}^{s}$ and $y_{2}^{s}$,
corresponding to the ports of the gate $s\in S$, then there are four
guard vertices $z_{jk}^{s}$ and eight edges $x_{j}^{s}\to z_{jk}^{s}$
and $z_{jk}^{s}\to y_{k}^{s}$, for $j,k\in\{1,2\}$. By removing
the guard vertices $z_{12}$ and $z_{21}$, this gadget behaves like
an inactive gate, while by removing $z_{11}$ and $z_{22}$ it behaves
like an active gate. The \emph{permutation gadget} graph $X$ of size
$n$ (Fig.~\ref{fig:perm-net}) contains the inlets $u_{1},\dots,u_{n}$
and the outlets $v_{1},\dots,v_{n}$ as vertices, a gate gadget for
each gate $s\in S$, and all the wires $W$ as extra edges. For a
given permutation $\pi$, define $K_{s}^{X}(\pi)=\{z_{11}^{s},z_{22}^{s}\}$
for $s\in A_{\pi}$, $K_{s}^{X}(\pi)=\{z_{12}^{s},z_{21}^{s}\}$ for
$s\in S\setminus A_{\pi}$ and $K^{X}(\pi)=\bigcup_{s\in S}K_{s}^{X}(\pi)$.
By removing all the vertices in $K^{X}(\pi)$, the gadget realizes
the permutation $\pi$: the only maximal paths in the graph not passing
through $K^{X}(\pi)$ are $P_{1},\dots,P_{n}$, where $P_{i}$ goes
from $u_{i}$ to $v_{\pi(i)}$. This completes the construction of
the gadget.

\begin{figure}
\begin{centering}
\includegraphics{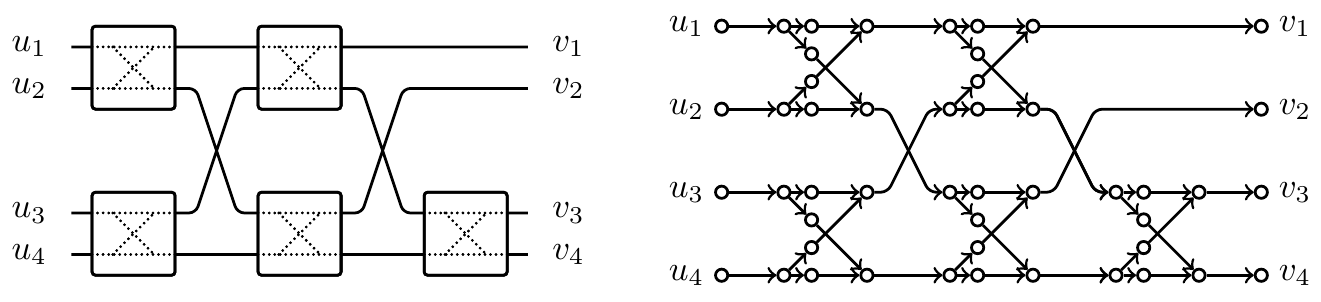}
\par\end{centering}

\centering{}\caption{\label{fig:perm-net}A permutation network of size $n=4$ and its
corresponding permutation gadget graph.}
\end{figure}

We now define a family of permutations $\pi_{i}$ indexed by $i\in\{1,\dots,n\}$.
Identify the $k$-th columns of $\mathbf{A}$ with the index $\ell=k$
and the $j$-th column of $\mathbf{C}$ with the index $\ell=m+j$.
For every $i$, let $\pi_{i}\colon\{1,\dots,2m\}\to\{1,\dots,2m\}$
be a permutation that sorts the indexes $\ell\in\{1,\dots,2m\}$ according
to the value $\mathbf{A}[i,k]$ for $\ell=k\in\{1,\dots,m\}$ and
$\mathbf{C}[i,j]$ for $\ell=m+j\in\{m+1,\dots,2m\}$, breaking ties
in favor of $\mathbf{A}$. Namely, $\pi_{i}$ is such that $\mathbf{A}[i,k]>\mathbf{C}[i,j]$
implies $\pi_{i}(k)>\pi_{i}(m+j)$ and $\mathbf{A}[i,k]\leq\mathbf{C}[i,j]$
implies $\pi_{i}(k)<\pi_{i}(m+j)$.

The graph $G_{1}$ is constructed as follows (see Fig.~\ref{fig:graph-bob}).
We start with four distinct vertices $x_{\ell},y_{\ell},z_{\ell},w_{\ell}$
for each $\ell\in\{1,\dots,2m\}$, forming four successive layers.
Between the first two layers, we add a permutation gadget $Y$ of
size $2m$, with inlets $x_{1},\dots,x_{2m}$ and outlets $y_{1},\dots,y_{2m}$,
associated with the permutation family $\pi_{i}^{Y}=\pi_{i}$. Then,
we add an edge $(y_{\ell},z_{\ell'})$ for each $\ell,\ell'\in\{1,\dots,2m\}$
with $\ell'\leq\ell$. Between the third and the fourth layers, we
add a permutation gadget $Z$ of size $2m$ with inlets $z_{1},\dots,z_{2m}$
and outlets $w_{1},\dots,w_{2m}$, associated with the family of inverse
permutations $\pi_{i}^{Z}=\pi_{i}^{-1}$. Finally, for each $j,k\in\{1,\dots,m\}$
such that $\mathbf{B}[k,j]=1$, we add a ``return'' edge $(w_{m+j},x_{k})$
connecting the last with the first layer. The set of final winning
configurations for Alice is ${\cal F}^{0}={\cal V}^{1}\cup\{\langle0,i,u\rangle\mid i\in\{0,\dots,n\}\text{ and }u\in K^{Y}(\pi_{i})\cup K^{Z}(\pi_{i}^{-1})\}$.

Our reduction is complete. We conclude the proof of Theorem~\ref{thm:msbmm-to-sim}
with Lemma~\ref{lem:triangle-implies-survive} and Lemma~\ref{lem:no-triangle-implies-win}.

\begin{figure}
\begin{centering}
\includegraphics{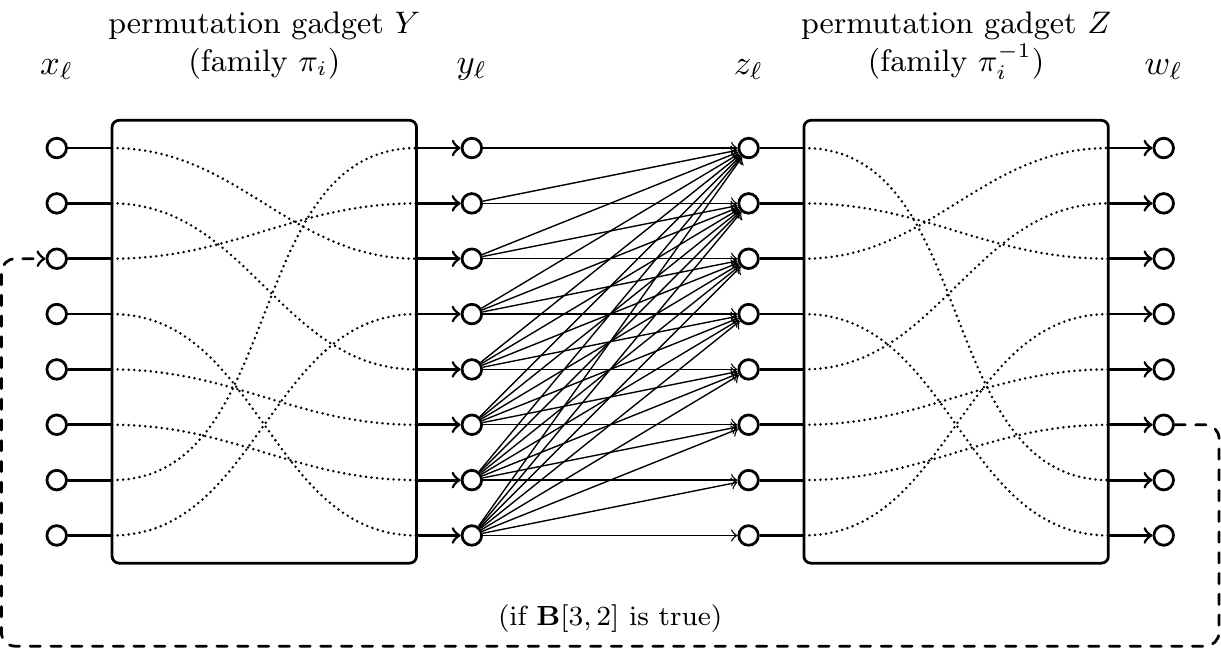}
\par\end{centering}

\centering{}\caption{\label{fig:graph-bob}A depiction of the graph $G_{1}$ for $m=4$.
The dotted lines within each permutation gadget $X\in\{Y,Z\}$ represent
maximal paths not passing through $K^{X}(\pi_{i}^{X})$, and depend
on the vertex $i$ in the graph $G_{0}$ of Alice where her token
is located. The dashed line is an example of edge $(w_{m+j},x_{k})$,
for $j=2$ and $k=3$, which is present only if $\mathbf{B}[3,2]=1$.
The actual graph has a similar edge for every $k$ and $j$ such that
$\mathbf{B}[k,j]=1$.}
\end{figure}

\begin{lem}
\label{lem:triangle-implies-survive}If $\mathbf{A},\mathbf{B},\mathbf{C}$
comprise an invalid triangle $(i,j,k)$, then $\langle0,i,x_{k}\rangle\in{\cal S}^{1}$.\end{lem}
\begin{proof}
First of all, Bob has to move his token following the permutation
networks $Y$ and $Z$. Namely, for $X\in\{Y,Z\}$, if his token is
on a non-last vertex $u$ in one of the maximal paths of $X$ not
passing through $K^{X}(\pi_{i}^{X})$, then Bob has to move his token
to the next vertex $u'$ in the path. This rule is necessary and sufficient
to ensure that the Alice cannot stop the play on ${\cal F}^{0}$ and
win.

To close a cycle, Bob moves first his token along the permutation
gadget $Y$, from $x_{k}$ to $y_{\pi_{i}(k)}$. Then, he moves it
from $y_{\pi_{i}(k)}$ to $z_{\pi_{i}(m+j)}$, which is possible since
$\mathbf{A}[i,k]>\mathbf{C}[i,j]$, so $\pi_{i}(k)>\pi_{i}(m+j)$
and the edge $(y_{\pi_{i}(k)},z_{\pi_{i}(m+j)})$ is given by construction.
Next, he follows the permutation gadget $Z$, moving the token from
$z_{\pi_{i}(m+j)}$ to $w_{\pi_{i}^{-1}(\pi_{i}(m+j))}=w_{m+j}$,
and, finally, he moves from $w_{m+k}$ back to $x_{j}$, which is
possible since $\mathbf{B}[k,j]=1$.\end{proof}
\begin{lem}
\label{lem:no-triangle-implies-win}If $\mathbf{A},\mathbf{B},\mathbf{C}$
do not comprise an invalid triangle, then ${\cal W}^{0}={\cal V}$.\end{lem}
\begin{proof}
The token of Alice starts on $i\in V_{0}$. If Bob moves his token
to a vertex in either $K^{Y}(\pi_{i})$ or $K^{Z}(\pi_{i}^{-1})$,
then Alice stops at the next turn and wins immediately. Since Bob
does not have at all the opportunity to stop and win the game (since\ ${\cal F}^{1}\cap{\cal V}^{1}=\emptyset$),
we are left with the question whether there is an infinite play in
which the token of Bob never passes through $K^{Y}(\pi_{i})$ or $K^{Z}(\pi_{i}^{-1})$.
We define a potential $p\colon V_{1}\to\mathbb{N}$ over the possible
locations of Bob's token, showing that it never increases along this
hypothetical infinite play, and strictly decreases frequently, a contradiction.

Let $P_{i,\ell}^{Y}$ be the only maximal path in $Y$ that goes from
$x_{\ell}$ to $y_{\pi_{i}(\ell)}$ and does not contain any vertex
in $K^{Y}(\pi_{i}^{Y})$. Similarly, for the second permutation gadget,
let $P_{i,\ell}^{Z}$ be the only maximal path in $Z$ that goes from
$z_{\pi_{i}(\ell)}$ to $w_{\ell}$ and does not contain any vertex
in $K^{Z}(\pi_{i}^{Z})$. For every vertex $u$ along the paths $P_{i,\ell}^{Y}$
and $P_{i,\ell}^{X}$ (including the endpoints $x_{\ell}$, $y_{\pi_{i}(\ell)}$,
$z_{\pi_{i}(\ell)}$ and $w_{\ell}$), let $p(u)=\pi_{i}(\ell)$.
The only possible moves are either
\begin{enumerate*}[label=(\alph*)]
\item along a path $P_{i,\ell}^{Y}$ or $P_{i,\ell}^{Z}$, where the potential
remains constant by definition, or
\item from $y_{\ell}$ to $z_{\ell'}$, where $p(y_{\ell})=\ell'\leq\ell=p(z_{\ell'})$,
or
\item \label{enu:return-edge}from $w_{m+j}$ to $x_{k}$.
\end{enumerate*}
In case \ref{enu:return-edge}, we have $p(w_{m+k})=\pi_{i}(m+k)$
and $p(x_{j})=\pi_{i}(j)$. Since there are no invalid triangles and
$\mathbf{B}[k,j]=1$, necessarily $\mathbf{A}[i,k]\leq\mathbf{C}[i,j]$,
so $\pi_{i}(m+k)<\pi_{i}(j)$, and the potential strictly decreases.
Furthermore, moves of type \ref{enu:return-edge} occur frequently
since without them the configuration graph is acyclic.
\end{proof}

\bibliographystyle{plain}
\bibliography{references}

\appendix

\section{Proof of properties of reachability games}

\label{sec:apx-games}
\begin{proof}[Proof of Proposition~\ref{prop:determinacy}]
Each property is proven below.
\begin{itemize}
\item[\ref{enu:stable}]  Take $s(\sigma)\in{\cal U}$ if $\sigma\in\lifts^{P}({\cal U})\cap{\cal V}^{P}\cap{\cal F}^{1-P}$
(by definition of $\lifts^{P}$, a move $(\sigma,\sigma')\in{\cal E}$
with $\sigma'\in{\cal U}$ exists) and $s(\sigma)=\bot$ otherwise.
Let $\pi=\sigma_{0}\sigma_{1}\cdots$ be a play starting with $\sigma_{0}\in{\cal U}$
and conforming to $s$. Then, the play~$\pi$ remains on ${\cal U}$,
and it is either winning for $P$ or infinite. Indeed, if $\sigma_{i}\in{\cal U}\subseteq\lifts^{P}({\cal U})$,
then either $\sigma_{i}$ is the final configuration in $\pi$ and
$\sigma_{i}\in{\cal F}^{P}$, or $\sigma_{i+1}\in{\cal U}$ for the
next configuration $\sigma_{i+1}$.
\item[\ref{enu:progress-measure}]  For $\sigma\in{\cal V}^{P}$, let $s(\sigma)=\bot$ if $p(\sigma)=0$,
and take $s(\sigma)$ such that $p(s(\sigma))<\liftp^{P}(p)(\sigma)\leq p(\sigma)$
if $1\leq p(\sigma)<\infty$ (such a $s(\sigma)$ exists by definition
of $\liftp^{P}$). For any play $\pi=\sigma_{0}\sigma_{1}\cdots$
conforming to $s$ with $p(\sigma_{0})<\infty$, the value $p(\sigma_{i})$
decreases strictly with $i$, so $\pi$ is finite. Let $\sigma_{\ell}$
be the last configuration of $\pi$. If $\sigma_{\ell}\in{\cal V}^{P}$,
then $p(\sigma_{\ell})=0$, otherwise $\sigma_{\ell}\in{\cal V}^{1-P}$
and $p(\sigma_{\ell})<\infty$: in either case $\sigma_{\ell}\in{\cal F}^{P}$
follows.
\item[\ref{enu:potential-fixpoint}]  Observe that $\liftp^{P}(p)$ is defined equivalently by the equations
$\{\sigma\mid\liftp^{P}(p)(\sigma)<k+1\}=\lifts^{P}(\{\sigma\mid p(\sigma)<k\})$
for $k\in\mathbb{N}$. Hence, construct recursively the sets ${\cal W}_{<k+1}^{P}=\lifts^{P}({\cal W}_{<k}^{P})$,
for every $k\in\mathbb{N}$, starting with ${\cal W}_{<0}^{P}=\emptyset$.
Clearly ${\cal W}_{<0}^{P}\subseteq{\cal W}_{<1}^{P}$, and by induction
${\cal W}_{<k}^{P}\subseteq{\cal W}_{<k+1}^{P}$ for every $k\in\mathbb{N}$
since $\lifts^{P}$ is monotone. Now define ${\cal W}_{k}^{P}={\cal W}_{<k+1}^{P}\setminus{\cal W}_{<k}^{P}$,
for $k\in\mathbb{N}$, observing that they are pairwise disjoint,
and let $r^{P}(\sigma)=k$, if $\sigma\in{\cal W}_{k}^{P}$ for some
$k$, or $r^{P}(\sigma)=\infty$ otherwise. The potential $r^{P}$
is the unique solution of the equation $\liftp^{P}(p)=p$ by construction.
Indeed, $\{\sigma\mid r^{P}(\sigma)<k+1\}=\lifts^{P}(\{\sigma\mid r^{P}(\sigma)<k\})$
for every $k\in\mathbb{N}$, so $\liftp^{P}(r^{P})=r^{P}$. Suppose
to have a distinct solution $\liftp^{P}(p)=p$ and take the smallest
$k$ such that $\{\sigma\mid p(\sigma)=k\}\neq{\cal W}_{k}^{P}$.
Then $\{\sigma\mid p(\sigma)<k\}={\cal W}_{<k}^{P}$ so $\{\sigma\mid p(\sigma)<k+1\}=\lifts^{P}(\{\sigma\mid p(\sigma)<k\})=\lifts^{P}({\cal W}_{<k}^{P})={\cal W}_{<k+1}^{P}$
and $\{\sigma\mid p(\sigma)=k\}=\{\sigma\mid p(\sigma)<k+1\}\setminus{\cal W}_{<k}^{P}={\cal W}_{k}^{P}$.
Since $r^{P}$ is a progress measure, $\supp(r^{P})\subseteq{\cal W}^{P}$.
Observe that $\lifts^{1-P}({\cal V}\setminus\supp(r^{P}))={\cal V}\setminus\lifts^{P}(\supp(r^{P}))={\cal V}\setminus\supp(r^{P})$,
so ${\cal V}\setminus\supp(r^{P})$ is closed for $1-P$ and contained
in ${\cal S}^{1-P}$. As ${\cal W}^{P}$ and ${\cal S}^{1-P}$ are
clearly disjoint, we get the statement.
\item[\ref{enu:set-fixpoint}]  We already showed that ${\cal W}^{P}$ and ${\cal S}^{P}$ are fixpoints
of $\lifts^{P}$. Being closed, any other fixpoint is contained in
${\cal S}^{P}$ by \ref{enu:stable}, so ${\cal S}^{P}$ is the greatest
fixpoint. By induction ${\cal W}_{<k}^{P}$ is contained in every
fixpoint, since ${\cal W}_{<0}^{P}=\emptyset$ and ${\cal W}_{<k+1}^{P}=\lifts^{P}({\cal W}_{<k}^{P})$.
Thus ${\cal W}^{P}$ is the least fixpoint.
\item[\ref{enu:algorithm}]  Let $\delta^{+}(\sigma)$ be the number of moves $(\sigma,\sigma')\in{\cal E}$,
and let $c_{k}(\sigma)\leq\delta^{+}(\sigma)$ be the number of moves
$(\sigma,\sigma')\in{\cal E}$ such that $\sigma'\in{\cal W}_{<k}^{P}$.
We can characterize ${\cal W}_{<k+1}^{P}$ as follows: for $\sigma\in{\cal V}_{P}$,
we have $\sigma\in{\cal W}_{<k+1}^{P}$ iff $\sigma\in{\cal F}^{P}$
or $c_{k}(\sigma)>0$, while for $\sigma\in{\cal V}_{1-P}$, we have
$\sigma\in{\cal W}_{<k+1}^{P}$ iff $\sigma\in{\cal F}^{P}$ and $c_{k}(\sigma)=\delta^{+}(\sigma)$.
Maintain a counter $c\colon{\cal V}\to\mathbb{N}$. Start with $c(\sigma)=c_{0}(\sigma)=0$
for every $\sigma\in{\cal V}$, and compute the set ${\cal W}_{0}^{P}=({\cal V}^{P}\cap{\cal F}^{P})\cup\{\sigma\in{\cal V}^{1-P}\cap{\cal F}^{P}\mid\delta^{+}(\sigma)=0\}$
in $O(|{\cal V}|)$ time. Then, for each $k=1,\dots,|{\cal V}|-1$,
compute $c_{k}$ and ${\cal W}_{k}^{P}$ as follows: for each move
$(\sigma,\sigma')\in{\cal E}$ with $\sigma'\in{\cal W}_{k-1}^{P}$,
increase the value of $c(\sigma)$ by one, so that at the end $c(\sigma)=c_{k}(\sigma)$
for every $\sigma\in{\cal V}$. If a configuration $\sigma$ satisfies
for the first time the condition $c_{k}(\sigma)>0$ (if $\sigma\in{\cal V}^{P}$)
or $c_{k}(\sigma)=\delta^{+}(\sigma)$ (if $\sigma\in{\cal V}^{1-P}\cap{\cal F}^{P}$),
then add $\sigma$ to ${\cal W}_{k}^{P}$. Since visiting a move is
done in constant time, and each move is visited at most once, this
phase requires $O({\cal E})$ time. The total time needed is then
$O(|{\cal V}|+|{\cal E}|)$.
\end{itemize}
\end{proof}

\section{Equivalence with classical definition of simulation preorder}

\label{sec:apx-simulation-classical}
\begin{defn}[Simulation preorder, classical definition]
\label{def:simulation-classical}A binary relation $R\subseteq S\times S$
is a \emph{simulation} if, for every $(s,t)\in R$, we have that (a)
$s$ and $t$ have the same label $L(s)=L(t)$, and (b) for every
transition $(s,s')\in T$ there is a transition $(t,t')\in T$ such
that $(s',t')\in R$.

For $s,t\in S$, we say that $t$ \emph{simulates} $s$ (written $s\preceq_{s}t$)
if there exists a simulation $R$ with $(s,t)\in R$.\end{defn}
\begin{prop}[\cite{Henzinger1995,Etessami2005,Cerny2010}]
\label{lem:sim-2pw}Definition~\ref{def:simulation-game} and Definition~\ref{def:simulation-classical}
of simulation preorder are equivalent.\end{prop}
\begin{proof}
($\implies$) Take a simulation relation $R\subseteq S\times S$ and
define 
\[
{\cal U}=\{\langle0,s,t\rangle\mid(s,t)\in R\}\cup\{\langle1,t,s'\rangle\mid\exists s\text{ such that }(s,t)\in R\text{ and }(s,s')\in T\}.
\]
We prove that ${\cal U}$ is closed on $({\cal G},{\cal F})$, so
${\cal U}\subseteq{\cal S}^{1}$. Take $\langle0,s,t\rangle\in{\cal U}$.
Since $(s,t)\in R$, we have $L(s)=L(t)$ by definition of simulation,
so $\langle0,s,t\rangle\in{\cal F}^{1}$. Moreover, for every $(\langle0,s,t\rangle,\langle1,t,s'\rangle)\in{\cal E}$
we have $(s,s')\in T$ so $\langle1,t,s'\rangle\in{\cal U}$ by construction.
Now take any $\langle1,t,s'\rangle\in{\cal U}$ and let $s\in S$
be such that $(s,t)\in R$ and $(s,s')\in T$. By definition of simulation,
there is a $t'\in S$ such that $(t,t')\in T$ and $(s',t')\in R$.
Hence, $(\langle1,t,s'\rangle,\langle0,s',t'\rangle)\in{\cal E}$
with $\langle0,s',t'\rangle\in{\cal U}$.

($\impliedby$) We prove that the relation $R=\{(s,t)\mid\langle0,s,t\rangle\in{\cal S}^{1}\}$
is a simulation. Suppose $(s,t)\in R$ so $\langle0,s,t\rangle\in{\cal S}^{1}$.
Observe that $\langle0,s,t\rangle\in{\cal F}^{1}$, so $L(s)=L(t)$,
otherwise Alice wins by stopping on $\langle0,s,t\rangle$. For any
edge $(s,s')\in T$, we have $(\langle0,s,t\rangle,\langle1,t,s'\rangle)\in{\cal E}$
and, since $\langle0,s,t\rangle\in{\cal S}^{1}$, also $\langle1,t,s'\rangle\in{\cal S}^{1}$.
However, since $\langle1,t,s'\rangle\notin{\cal F}^{1}$, then there
exists a $t'\in S$ such that $(\langle1,t,s'\rangle,\langle0,s',t'\rangle)\in{\cal E}$
and $\langle0,s',t'\rangle\in{\cal S}^{1}$. In particular, $(t,t')\in T$
and $(s',t')\in R$.
\end{proof}

\section{Equivalence of simulation games and 2TRGs}

\label{sec:apx-2trg-sim-equiv}
\begin{proof}[Proof of Theorem~\ref{thm:2trg-sim-equiv}, continues]
We need to prove that ${\cal S}_{{\cal G},{\cal F}}^{1}={\cal S}_{{\cal G}',{\cal F}'}^{1}$.

($\supseteq$) Define the potential $p$ on ${\cal G}'$ as follows
\[
\begin{aligned}p(\langle P,u,v\rangle) & =r_{{\cal G},{\cal F}}^{0}(\langle P,u,v\rangle)+1 &  & \text{for }\langle P,u,v\rangle\in{\cal V}\\
p(\langle1,v,u^{*}\rangle) & =0 &  & \text{for }\langle0,u,v\rangle\in{\cal F}^{0}\\
p(\langle P,x,y\rangle) & =\infty &  & \text{in any other case}
\end{aligned}
.
\]
Observe that $p$ is a progress measure on $({\cal G}',{\cal F}')$
for Alice. Thus ${\cal W}_{{\cal G},{\cal F}}^{0}=\supp(p)\cap{\cal V}\subseteq\supp(p)\subseteq{\cal W}_{{\cal G}',{\cal F}'}^{0}$.

($\subseteq$) Define the set ${\cal U}\subseteq{\cal V}_{{\cal G}'}$
as follows
\[
\begin{aligned}\langle P,u,v\rangle & \in{\cal U} &  & \text{for }\langle P,u,v\rangle\in{\cal S}_{{\cal G},{\cal F}}^{1}\\
\langle1,v,u^{*}\rangle & \in{\cal U} &  & \text{for }\langle0,u,v\rangle\in{\cal F}^{1}\\
\langle0,u,u\rangle,\langle0,u^{*},u^{*}\rangle & \in{\cal U} &  & \text{for }u\in V_{0}\\
\langle1,u,u'\rangle & \in{\cal U} &  & \text{for }(u,u')\in E_{0}\\
\langle P,x,y\rangle & \notin{\cal U} &  & \text{in any other case}
\end{aligned}
.
\]
Observe that ${\cal U}$ is closed on $({\cal G}',{\cal F}')$ for
Bob. Thus ${\cal S}_{{\cal G},{\cal F}}^{1}={\cal U}\cap{\cal V}\subseteq{\cal U}\subseteq{\cal S}_{{\cal G}',{\cal F}'}^{1}$.
\end{proof}

\section{Proof of dicut decomposition of reachability games}

\label{sec:apx-dicut}
\begin{proof}[Proof of Lemma~\ref{lem:dicut}, continues]

Denote with $\lifts_{H}^{P}$ and $\lifts_{T}^{P}$ the operators
$\lifts^{P}$ on the games $({\cal G}[{\cal V}_{H}],{\cal F}[{\cal V}_{H}])$
and $({\cal G}[{\cal V}_{T}],{\cal F}_{*}[{\cal V}_{H}])$. We first
need to show that $\lifts_{H}^{P}({\cal U}\cap{\cal V}_{H})=\lifts^{P}({\cal U})\cap{\cal V}_{H}$,
for any ${\cal U}\subseteq{\cal V}$, and that $\lifts_{T}^{P}({\cal U}_{T})=\lifts^{P}({\cal U}_{T}\cup{\cal S}_{H}^{P})\cap{\cal V}_{T}$
for any ${\cal U}_{T}\subseteq{\cal V}_{T}$.

For ${\cal U}\subseteq{\cal V}$ and $\sigma\in{\cal V}_{H}\cap{\cal V}^{P}$
we have

\begin{eqnarray*}
\sigma\in\lifts_{H}^{P}({\cal U}\cap{\cal V}_{H}) & \iff & \sigma\in{\cal F}^{P}\lor{\textstyle \bigvee_{(\sigma,\sigma')\in{\cal E}}}\:\sigma'\in{\cal U}\cap{\cal V}_{H}\\
 & \iff & \sigma\in\lifts^{P}({\cal U})
\end{eqnarray*}
so $\lifts_{H}^{P}({\cal U}\cap{\cal V}_{H})\cap{\cal V}^{P}=\lifts^{P}({\cal U})\cap{\cal V}_{H}\cap{\cal V}^{P}$.
We obtain $\lifts_{H}^{P}({\cal U}\cap{\cal V}_{H})=\lifts^{P}({\cal U})\cap{\cal V}_{H}$
by applying de Morgan laws.

For ${\cal U}_{T}\subseteq{\cal V}_{T}$ and $\sigma\in{\cal V}_{T}\cap{\cal V}^{P}$
we have 
\begin{eqnarray*}
\sigma\in\lifts_{T}^{P}({\cal U}_{T}) & \iff & \sigma\in{\cal F}_{*}^{P}\;\lor\;{\textstyle \bigvee_{(\sigma,\sigma')\in{\cal E}}}\:\sigma'\in{\cal U}_{T}\\
 & \iff & \sigma\in\lifts^{P}({\cal S}_{H}^{P})\;\lor\;{\textstyle \bigvee_{(\sigma,\sigma')\in{\cal E}}}\:\sigma'\in{\cal U}_{T}\\
 & \iff & \sigma\in{\cal F}^{P}\;\lor\;{\textstyle \bigvee_{(\sigma,\sigma')\in{\cal E}}}\:\sigma'\in{\cal S}_{H}^{P}\;\lor\;{\textstyle \bigvee_{(\sigma,\sigma')\in{\cal E}}}\:\sigma'\in{\cal U}_{T}\\
 & \iff & \sigma\in{\cal F}^{P}\;\lor\;{\textstyle \bigvee_{(\sigma,\sigma')\in{\cal E}}}\:\sigma'\in{\cal U}_{T}\cup{\cal S}_{H}^{P}\\
 & \iff & \sigma\in\lifts^{P}({\cal U}_{T}\cup{\cal S}_{H}^{P})
\end{eqnarray*}
so $\lifts_{T}^{P}({\cal U}_{T})\cap{\cal V}^{P}=\lifts^{P}({\cal U}_{T}\cup{\cal S}_{H}^{P})\cap{\cal V}_{T}\cap{\cal V}^{P}$.
We obtain $\lifts_{T}^{P}({\cal U}_{T})=\lifts^{P}({\cal U}_{T}\cup{\cal S}_{H}^{P})\cap{\cal V}_{T}$
by applying de Morgan laws. 
\begin{enumerate}
\item ${\cal S}_{H}^{P}\cup{\cal S}_{T}^{P}$ is closed on $({\cal G},{\cal F})$:
\begin{eqnarray*}
{\cal S}_{H}^{P}\cup{\cal S}_{T}^{P} & = & \lifts_{H}^{P}({\cal S}_{H}^{P})\cup\lifts_{T}^{P}({\cal S}_{T}^{P})\\
 & = & \left[\lifts^{P}({\cal S}_{H}^{P})\cap{\cal V}_{H}\right]\cup\left[\lifts^{P}({\cal S}_{T}^{P}\cup{\cal S}_{H}^{P})\cap{\cal V}_{T}\right]\\
 & \subseteq & \lifts^{P}({\cal S}_{T}^{P}\cup{\cal S}_{H}^{P})
\end{eqnarray*}
so ${\cal S}_{H}^{P}\cup{\cal S}_{T}^{P}\subseteq{\cal S}^{P}$,
\item ${\cal S}^{P}\cap{\cal V}_{H}$ is closed on $({\cal G}_{H},{\cal F}_{H})$:
\begin{eqnarray*}
{\cal S}^{P}\cap{\cal V}_{H} & = & \lifts^{P}({\cal S}^{P})\cap{\cal V}_{H}\\
 & = & \lifts_{H}^{P}({\cal S}^{P}\cap{\cal V}_{H})
\end{eqnarray*}
so ${\cal S}^{P}\cap{\cal V}_{H}\subseteq{\cal S}_{H}^{P}$ and, together
with (1), ${\cal S}^{P}\cap{\cal V}_{H}={\cal S}_{H}^{P}$,
\item ${\cal S}^{P}\cap{\cal V}_{T}$ is closed on $({\cal G}_{T},{\cal F}_{T})$:
\begin{eqnarray*}
{\cal S}^{P}\cap{\cal V}_{T} & = & \lifts^{P}({\cal S}^{P})\cap{\cal V}_{T}\\
 & = & \lifts^{P}({\cal S}_{H}^{P}\cup({\cal S}^{P}\cap{\cal V}_{T}))\cap{\cal V}_{T}\\
 & = & \lifts_{T}^{P}({\cal S}^{P}\cap{\cal V}_{T})
\end{eqnarray*}
so ${\cal S}^{P}\cap{\cal V}_{T}\subseteq{\cal S}_{T}^{P}$ and, together
with (1), ${\cal S}^{P}\cap{\cal V}_{T}={\cal S}_{T}^{P}$.\end{enumerate}
\end{proof}

\end{document}